\documentclass[letterpaper, 10 pt,journal]{ieeetran}  

\IEEEoverridecommandlockouts                              
\usepackage{amsmath,amstext}
\usepackage{amsthm}
\usepackage{amssymb}
\usepackage{physics}

\usepackage{lscape}
\usepackage{algorithm}
\usepackage{algpseudocode}
\usepackage[colorinlistoftodos]{todonotes}
\usepackage{mathtools}
\usepackage{pgfplots} 
\pgfplotsset{compat=newest} 
\pgfplotsset{compat=newest}
\pgfplotsset{plot coordinates/math parser=false}

\usepackage{refcount}
\usepackage[normalem]{ulem}

\usepackage[flushleft]{threeparttable}

\colorlet{lcolor}{blue!40!black}
\colorlet{ucolor}{magenta!40!black}
\colorlet{ccolor}{green!40!black}

\usepackage[colorlinks=true,%
linkcolor=lcolor,%
urlcolor=ucolor,%
citecolor=ccolor]{hyperref}

\theoremstyle{definition}
 \newtheorem{asum}{Assumption}
 \newtheorem{theo}{Theorem}
 \newtheorem{rema}{Remark}
 
 \newtheorem{lemma}{Lemma}
 \newtheorem{prop}{Proposition}
 \newtheorem{coro}{Corollary}

 \newcommand{\CLF}{\gamma}
  
  \newcommand{\boundd}{{b}}
  \newcommand{\change}[2]{#2}
 \newcommand{\MH}[1]{}
 \newcommand{\partder}[1]{V_\CLF^{'}(#1) }
 \newcommand{\timeder}[2]{\mathcal{L}_fV_\CLF(#1,#2)}

\newcommand{\refalgoo}{\ref{algo_cont_dyn_trig}} 
\newcommand{\tone}{r}
\newcommand{\arxiv}[2]{#2}

\usepackage{graphicx}

\arxiv{}{
\pubid{\begin{minipage}{\textwidth}\ \\[12pt]
		\copyright 2019 IEEE.  Personal use of this material is permitted.  Permission from IEEE must be obtained for all other uses, in any current or future media, including reprinting/republishing this material for advertising or promotional purposes, creating new collective works, for resale or redistribution to servers or lists, or reuse of any copyrighted component of this work in other works.
	\end{minipage}}}

\title{\LARGE \bf
Nonlinear Dynamic Periodic Event-Triggered Control with Robustness to Packet Loss Based on Non-Monotonic Lyapunov Functions
}

\author{Michael Hertneck, Steffen Linsenmayer, Frank Allg\"ower
	\thanks{The authors thank the German Research Foundation (DFG) for support of
		this work within the German Excellence
		Strategy under grant EXC-2075.}
	\thanks{The authors are with the Institute for Systems Theory and Automatic Control, University of Stuttgart, 70569 Stuttgart, Germany (email: $\{$hertneck,linsenmayer,allgower$\}$@ist.uni-stuttgart.de).}%
} 

\begin{document}


\maketitle

\arxiv{
	\thispagestyle{empty}
	\pagestyle{empty}
	}{}

\begin{abstract}
This paper considers the stabilization of nonlinear continuous-time dynamical systems employing periodic event-triggered control (PETC). Assuming knowledge of a stabilizing feedback law for the continuous-time system with a certain convergence rate, a dynamic, state dependent PETC mechanism is designed. The proposed mechanism guarantees on average the same worst case convergence behavior except for tunable deviations. Furthermore, a new approach to determine the sampling period for the proposed PETC mechanism is presented. This approach as well as the actual trigger rule exploit the theory of non-monotonic Lyapunov functions.  An additional feature of the proposed PETC mechanism is the possibility to integrate knowledge about packet losses in the PETC design. The proposed PETC mechanism is illustrated with a nonlinear numerical example from literature. \arxiv{}{This paper is the accepted version of \cite{hertneck2019nonlinear}, containing also the proofs of the main results.}



\end{abstract} 

\section{Introduction}
\label{sec_intro}
Networked control systems (NCS) are control systems in which some or all links in the feedback loop are replaced by a shared communication network. Whilst NCS are useful in many modern control applications, several network induced problems have to be addressed (for a detailed overview see e.g. \cite{Hespanha2007}). One major challenge in the field of NCS is the design of sampling and control strategies that use the network as little as possible to keep it available for other applications while being robust to unavoidable network induced imperfections as e.g packet loss. Nevertheless, stability and performance goals like a certain convergence rate of the system state need to be guaranteed.
A huge step towards tackling these two conflicting objectives was made by the development of event-triggered control (ETC) paradigms \cite{Heemels2012}.

In ETC, control updates are sent over the network according to a system state dependent trigger rule. While classical results on ETC employ a static trigger rule \cite{tabuada2007event,ong2018event}, in \cite{girard2015dynamic} the concept of dynamic ETC, where the trigger rule changes dynamically over time, has been introduced. 
 However, ETC approaches like those in \cite{tabuada2007event}-\cite{girard2015dynamic} require the continuous evaluation of the trigger rule, which makes their implementation on digital platforms impossible. In periodic event-triggered control (PETC) \cite{Heemels2013,heemels2013periodic}, this problem is overcome by evaluating the trigger rule only periodically at fixed sampling times.  New information is  transmitted at a sampling time if the trigger rule indicates it. Anyhow, whilst the number of transmissions can still be reduced in comparison to time-triggered control, stability and performance guarantees from ETC are in general not preserved for PETC.  Thus, it is desirable to find mechanisms tailored for PETC, such that stability, robustness to network imperfections and performance goals as a certain convergence rate of the system state can be guaranteed. Especially for nonlinear system dynamics, the design of such PETC mechanisms is a challenging task and deserves a comprehensive investigation.

There exists a bunch of results considering PETC for linear systems, see e.g. \cite{Heemels2013,heemels2013periodic} for an overview.  In dynamic PETC, a dynamically changing trigger rule is used like in dynamic ETC. 
	In \cite{linsenmayer2018event2}, dynamic PETC for linear systems with robustness to packet loss is investigated. 
 PETC results for nonlinear systems, either static or dynamic, are more rare. 
 In \cite{Borgers2018a}, an ETC trigger rule is overapproximated to obtain a PETC trigger rule. In \cite{Postoyan2013,wang2019periodic}, a PETC is emulated based on a stabilizing continuous-time controller. In \cite{wang2019periodic}, stability guarantees rely on the existence of a \change{special}{} hybrid Lyapunov function. For specific classes of nonlinear systems, PETC is investigated in \cite{Etienne2017,xu2018periodic}. In \cite{proskurnikov2018lyapunov} it is shown, that ETC and PETC with a known (and chosen) convergence rate for nonlinear systems can be designed, provided that a control Lyapunov function  is known for the continuous-time system. 
 
\arxiv{}{\pubidadjcol}

Even  though first results on PETC of nonlinear systems are available, research is still at an early stage and it is thus desirable to find improved PETC mechanisms. It is furthermore worthwhile to prolongate the maximum admissible sampling period (MASP) of the PETC, which plays an important role for reducing the number of transmissions over the network.
Also, the influence of network imperfections like packet loss on the PETC deserves a thorough investigation, in order to deal with aspects that arise in real world applications of NCS.

In this paper, we present a novel dynamic trigger mechanism for PETC of nonlinear continuous-time systems that can guarantee stability and a chosen averaged convergence rate if a controller and a Lyapunov function for the continuous-time system are known. The proposed trigger mechanism is based on non-monotonic Lyapunov functions and can be viewed as a nonlinear counterpart to the dynamic trigger mechanism from \cite{linsenmayer2018event2}. It can be applied to a wide class of nonlinear systems  
 and is robust to random packet loss if a bound on the number of successive lost packets is known.
 
  A lower bound on the MASP for the novel trigger mechanism is constructed based on an extension of results from \cite{proskurnikov2018lyapunov} considering non-monotonic Lyapunov functions \cite{michel2015stability} such that stability and a certain averaged convergence rate (here  indicated  by a parameter $\sigma$) can be guaranteed. This $\sigma$-MASP bound  depends on level sets of the considered continuous-time Lyapunov function and is increased in most cases by factor  $\frac{9}{4}$ in comparison to the $\sigma$-MASP for the PETC from \cite{proskurnikov2018lyapunov} while still guaranteeing the same averaged worst case convergence rate except a time shift of $(m+1)$ sampling periods of the PETC, if the number of successive lost packets is bounded by $m$. 

\arxiv{}{This paper is the accepted version of \cite{hertneck2019nonlinear}, containing also the proofs of our main results.}

The remainder of this paper is structured as follows. The problem setup is described in Section~\ref{sec_setup}. Some basic results from~\cite{proskurnikov2018lyapunov} and \cite{michel2015stability} are recapped in Section~\ref{sec_basic}. The improved bound on the $\sigma$-MASP and the dynamic trigger mechanism are presented in Section~\ref{sec_main_res}. A numerical example to illustrate the proposed PETC mechanism is given in Section~\ref{sec_example} and Section~\ref{sec_conclusion} concludes the paper. \arxiv{}{Some spacious proofs are given in the Appendix.}

 \subsubsection*{Notation}
The positive (respectively nonnegative) real numbers are denoted by $\mathbb{R}_{>0}$,  respectively  $\mathbb{R}_{\geq 0} = \mathbb{R}_{> 0} \cup \{0\} $. The positive (respectively nonnegative) natural numbers are denoted by $\mathbb{N}$, respectively $\mathbb{N}_0:=\mathbb{N}\cup  \left\lbrace 0 \right\rbrace $. A continuous function $\alpha: \mathbb{R}_{\geq 0} \rightarrow \mathbb{R}_{\geq 0}$ is a class $ \mathcal{K}$ function ,i.e., $\alpha \in \mathcal{K}$), if $\alpha$ is strictly increasing and $\alpha(0) = 0$. The notation $t^-$ is used as $t^- := \lim\limits_{s<t,s\rightarrow t} s$. A continuous function $V:\mathbb{R}^n \rightarrow \mathbb{R}$ is positive definite if $V(0) = 0$ and $V(x)>0$ for all  $x \neq 0$.  $V'(k)$ denotes $\left. \frac{\partial V(x)}{\partial x} \right|_{x = k}$.  Furthermore, we use in a slight abuse of notation $\mathcal{L}_fV(x,u)$ to denote the Lie derivative of $V$ along the vector field  $f:\mathbb{R}^n \times \mathbb{R}^b \rightarrow \mathbb{R}^n$, i.e. $\mathcal{L}_fV(x,u) = V^{'}(x)  f(x,u) $.

\section{Problem Setup}
\label{sec_setup}
In this section, we present the setup of this paper and formalize the control objective. 
\change{We introduce the considered setup for the NCS,  model the network, present the triggering strategy  and state a convergence criterion for the PETC.}{}
\subsection{Basic Setup}
	We consider a nonlinear, time-invariant system
	\begin{equation}
	\label{eq_sys_cont}
		\dot{x} = f(x,u)
	\end{equation}
	with a smooth vector valued function $f:\mathbb{R}^{n}\times\mathbb{R}^{b} \rightarrow \mathbb{R}^n$ satisfying $f(0,0) = 0$, 
	the system state $ x(t)\in \mathbb{R}^n$ with initial condition $x(0) = x_0$ and the input $u(t)\in\mathbb{R}^b$. The input is generated by 
	\begin{equation}
	\label{eq_cont_fb}
	u = \kappa(\hat{x})
	\end{equation} 
	with the nonlinear feedback law $\kappa:\mathbb{R}^n\rightarrow\mathbb{R}^b$ and a prediction $\hat{x}(t)$ of the system state $x(t)$ that is generated at the actuator based on transmitted state information and $\hat{x}(0) = \hat{x}_0$. The time instants, when state information is received by the actuator are given by the infinite sequence $(\tau_k)_{k \in \mathbb{N}_0}$  and define a discrete set 
\begin{equation*}
		\mathcal{T} := \left\lbrace \tau_0, \tau_1,\tau_2,\dots \right\rbrace.
	\end{equation*}
	The sequence $\mathcal{T}$ depends on a trigger mechanism, that will be designed in this paper, and on the capabilities of the communication network. However, we assume that current state information is received successfully at $t=0$ and thus have  $\tau_0 = 0$ and $\hat{x}_0 = x_0$. The update of $\hat{x}$ at $t \in \mathcal{T}$ is represented by $\hat{x}(t)  = x(t)$.  Between the update times, a state prediction can be designed based on the computational capabilities of the actuator as 
	\begin{equation}
	\label{eq_cont_est}
		\dot{\hat{x}}(t) = f_e(\hat{x}(t)),~  t \notin \mathcal{T}
	\end{equation} 
	 with $f_e(0) = 0$. If there are no computational capabilities, we can choose $f_e(x) = 0$, 
	 which corresponds to the zero order hold (ZOH) case.  In this case, the next input $\hat{u} = \kappa(x)$ can be transmitted instead of the system state $x$ even though we subsequently model the general case with $\hat{x}$ as a state. 

	 The closed-loop system combined of \eqref{eq_sys_cont}, controller \eqref{eq_cont_fb}, prediction \eqref{eq_cont_est}
	and its reset condition can be described as a discontinuous dynamical system (DDS) with state $\xi = \left[\xi_1^\top \xi_2^\top \right]^\top =  \left[x^\top \hat{x}^\top\right]^\top$ as 
	 \begin{align}
	 \dot{\xi}(t) =& 
	 \begin{bmatrix}
	 f(\xi_1(t),\kappa(\xi_2(t))) \\
	 f_{e}(\xi_2(t))
	 \end{bmatrix},~t \notin \mathcal{T}
	 , \nonumber \\
	 \xi(t) =& \begin{bmatrix}
	 I\\
	 I
	 \end{bmatrix} \xi_1(t^-), t \in \mathcal{T}\backslash \left\lbrace 0 \right\rbrace,
	 \label{eq_cont_sys_comp}
	 \end{align}
	 where $\xi(0) = \left[x_0^\top \hat{x}_0^\top\right]^\top = \xi_0$.
	 In order to design the PETC, we assume that a continuous-time feedback \change{law }{}and a Lyapunov function are known, \change{that satisfy}{satisfying} the following assumption.
 	\begin{asum} (cf. \cite{proskurnikov2018lyapunov})
		\label{asum_cont_clf}
		There is a continuous, positive definite function $V_\CLF:\mathbb{R}^n \rightarrow \mathbb{R}$, satisfying 
		\begin{equation}
			\label{eq_cont_clf2}
			\alpha_1 {(\norm{x})} \leq V_\CLF(x) \leq \alpha_2 {(\norm{x})}, 	
		\end{equation}	
		\begin{equation}
			\label{eq_cont_clf}
			V_\CLF^{'}(x(t))  f(x(t),\kappa(x(t))) \leq - \gamma(V_\CLF(x(t))).
		\end{equation}
		with class $\mathcal{K}$ functions $\gamma, \alpha_1, \alpha_2 $.
	\end{asum}
	Finding $\kappa$ and $V_\CLF$ that satisfy Assumption~\ref{asum_cont_clf} is a fundamental problem in control theory for continuous-time systems and is widely discussed in literature, see e.g.  \cite{khalil2002nonlinear}. Thus we will not review it here in more detail.
	
	Subsequently, we consider local results for a level set of $V_\CLF$, \change{that can be }{} defined as $\mathcal{X}_c := \left\lbrace x | V_\CLF(x) \leq c \right\rbrace$ for a chosen $c \in 
	\mathbb{R}_{>0}$.

\subsection{Network Model and Triggering Strategy}
We consider  an unreliable network that can transmit packets periodically with the sampling period $ h \in \mathbb{R}_{>0}$. 
 Network imperfections are modeled using the following assumption.
\begin{asum}
	\label{asum_loss_bound}
	The consecutive number of lost packets is bounded by $m\in\mathbb{N}$ and there is an acknowledgment if a transmission was successful.
\end{asum}
This assumption allows random packet dropouts and requires no knowledge about the underlying probability distribution as long as the boundedness condition is satisfied. It resembles the scenario where messages are dropped if they have a delay that is not negligible when the occurrence of such delays can be limited to a bounded number of successive transmissions.

 Necessary transmissions of the system state are detected using PETC, i.e. according to a trigger mechanism that is evaluated at discrete, evenly distributed time instants with the sampling period $ h $. The trigger rule of the PETC mechanism is thus evaluated at sampling times $t = z h$ for all $z\in\mathbb{N}_0$. If the trigger rule of the PETC mechanism is violated at a sampling time, a transmission of the system state is triggered. 

\subsection{Convergence Criterion and Control Objective}
A common convergence criterion based on $V_\CLF$ and $\gamma$ from Assumption~\ref{asum_cont_clf}, that is used e.g. in \cite{proskurnikov2018lyapunov}, is
\begin{equation}
\label{eq_cont_cont_conv}
\frac{d}{dt} V_\CLF(x(t)) \leq -\sigma \gamma(V_\CLF(x(t))),~\forall t \geq 0
\end{equation}
for some $\sigma \in \left(0,1\right)$. We note that if \eqref{eq_cont_cont_conv} holds, then it holds due to the comparison Lemma \cite[pp. 102-103]{khalil2002nonlinear} that $V_\CLF(x(t)) \leq S(t,x_0)$, where $S(t,x_0)$ is the solution of 
\begin{equation}
\label{eq_cont_def_s}
\frac{d}{dt} S(t,x_0) = -\sigma \gamma(S(t,x_0)), ~S(0,x_0) = V_\CLF(x_0).
\end{equation}
Thus $S(t,x_0)$ describes the worst case convergence behavior for \eqref{eq_cont_cont_conv} and can be used as a convergence criterion as e.g. discussed in \cite{ong2018event}, where an ETC is designed with a performance barrier based on $S(t,x_0)$. 
For our PETC mechanism, we use an averaged criterion similar to \eqref{eq_cont_cont_conv} that can be described as
\begin{equation}
	V_\CLF(x(t+(m+1)h)) \leq S(t,x_0)~ \forall t \geq 0\label{eq_cont_v_smaller}.
\end{equation}
Thus, if we use \eqref{eq_cont_v_smaller} as convergence criterion, then we require the same averaged worst case convergence rate as if we consider \eqref{eq_cont_cont_conv} except for a time shift depending on the sampling rate $h$ and the bound on successive lost packets $m$ that is small if $m$ and $h$ are small. We define the maximum admissible sampling period such that \eqref{eq_cont_v_smaller} can be guaranteed  as $\sigma$-MASP.

The goal of this paper is to find a lower bound on the $\sigma$-MASP that can be used to determine $h$ and to find a corresponding PETC mechanism, both such that asymptotic stability of the origin of the DDS~\eqref{eq_cont_sys_comp} is guaranteed and the convergence criterion~\eqref{eq_cont_v_smaller} is satisfied for all initial conditions from the level set $\mathcal{X}_c$. 

\section{Basic Results}
\label{sec_basic}
Before we present our main results, i.e. how the $\sigma$-MASP bound and the trigger mechanism can be constructed, we recap some results from literature that are important ingredients for the proposed PETC approach. First, we present a sufficient local stability condition for the DDS~\eqref{eq_cont_sys_comp} based on non-monotonic Lyapunov functions, that is a special case of Theorem~6.4.2 from~\cite{michel2015stability}. 
\change{This result can be used for guaranteeing asymptotic stability for the DDS~\eqref{eq_cont_sys_comp} controlled with the proposed PETC mechanism. }{}
In the second subsection, we state a sufficient condition for the convergence criterion \eqref{eq_cont_v_smaller} that is easier to verify than the criterion itself. In the last subsection, we recap a technical result and a set of assumptions from \cite{proskurnikov2018lyapunov}, which will be useful for designing the PETC mechanism and for constructing a lower bound on the $\sigma$-MASP.

\subsection{Non-Monotonic Stability Results for Event-Triggered NCS}
\change{Now, we present the sufficient stability condition for the DDS model~\eqref{eq_cont_sys_comp}  that is a special case of Theorem~6.4.2 from~\cite{michel2015stability}. It can be formulated as the following proposition.}{The stability condition for the DDS model~\eqref{eq_cont_sys_comp}  based on Theorem~6.4.2 from~\cite{michel2015stability} can be formulated as follows.}

\begin{prop}
	\label{prob_michel}
	Observe the DDS given by \eqref{eq_cont_sys_comp}. Assume that the unbounded discrete subset $\mathcal{T}$ of $\mathbb{R}_{\geq 0}$ satisfies 
	\begin{equation}
	\label{eq_non_mon_dec3}
		0 < \underline{\eta} \leq \tau_{k+1} - \tau_{k} \leq \overline{\eta}~ \forall k\in\mathbb{N}_0
	\end{equation}
	and $\tau_0 = 0$. Furthermore, assume there is a continuous positive definite function $V:\mathbb{R}^{2n} \rightarrow \mathbb{R},$ such that for all $k\in \mathbb{N}_0,$ and all $\xi(\tau_k) \in \mathcal{X}_{c,2} $, where $\mathcal{X}_{c,2} := \left\lbrace \xi|V(\xi) \leq c \right\rbrace$ for the chosen $c \in \mathbb{R}_{>0}$ and  class $\mathcal{K}$ functions $\alpha_3, \alpha_4, \gamma_2 $,
	\begin{equation}
	\label{eq_non_mon_dec4}
		\alpha_3 {(\norm{\xi})} \leq V(\xi) \leq \alpha_4 {(\norm{\xi})} 	
	\end{equation}
	\begin{equation}
	\label{eq_non_mon_desc1}
		V(\xi(\tau_k+r)) \leq V(\xi(\tau_k)), ~ 0 \leq r \leq \tau_{k+1} -\tau_k
	\end{equation}
	and
	\begin{equation}
		\label{eq_non_mon_desc2}
		\frac{1}{\tau_{k+1} - \tau_k} \left[ V(\xi(\tau_{k+1})) - V(\xi(\tau_k))\right] \leq -\gamma_2 (V(\xi(\tau_k)))
	\end{equation}
	hold.
	Then the equilibrium $\xi = 0$ is asymptotically stable for \eqref{eq_cont_sys_comp}  with region of attraction $\mathcal{X}_{c,2}$. 
\end{prop}
\begin{proof}
	Follows \change{as a special case of}{from} Theorem~6.4.2 from \cite{michel2015stability}. 
\end{proof}
Thus, if the Lyapunov function $V(\xi(t))$ decreases along the sequence $(\tau_k)_{k\in \mathbb{N}_0}$ and is bounded between successive times $\tau_k$ and $\tau_{k+1}$ from $\mathcal{T}$ by the value at the last successful transmission, i.e. $V(\xi(\tau_k))$, and if in addition the time between successive instants from $\mathcal{T}$ is uniformly lower and upper bounded, then asymptotic stability follows. We use this later in order to prove stability for the DDS~\eqref{eq_cont_sys_comp}, controlled with the proposed PETC mechanism. Proposition~\ref{prob_michel} implies also the existence and uniqueness of solutions for the DDS~\eqref{eq_cont_sys_comp}, for details see \cite{michel2015stability}. Moreover, invariance of $\mathcal{X}_{c,2}$ is a direct consequence of \eqref{eq_non_mon_desc1} and $\tau_0 = 0$. Henceforth, we will consider the ZOH case, i.e., $f_e(x) = 0$ for which the conditions from Proposition~\ref{prob_michel} can be simplified as follows.

\begin{prop}
	\label{prop_help_eq}
		Let Assumption~\ref{asum_cont_clf} hold and assume $f_e(x) = 0 $. Then, \eqref{eq_non_mon_dec4}, \eqref{eq_non_mon_desc1} and \eqref{eq_non_mon_desc2} hold for the DDS~\eqref{eq_cont_sys_comp} and $V(\xi) = \frac{1}{2} \left(V_\CLF(\xi_1) + V_\CLF(\xi_2)\right)$, if
		\begin{equation}
		\label{eq_v_desc}
		V_\CLF(x(\tau_k+\tone)) \leq V_\CLF(x(\tau_k)) 
		\end{equation} 
		holds for $0\leq \tone < \tau_{k+1}-\tau_k$, and all $k \in \mathbb{N}_0$ on $\mathcal{X}_c$ and 
		\begin{align}
		\nonumber
		&\frac{1}{\tau_{k+1} - \tau_k} \left[ V_\CLF(x(\tau_{k+1})) - V_\CLF(x(\tau_k))\right]\\
		\leq& -\gamma_2 (V_\CLF(x(\tau_k))) \label{eq_v_desc_2}
		\end{align}
		holds for all $k\in\mathbb{N}_0$ on $\mathcal{X}_c$. Furthermore, $x(\tau_k) \in\mathcal{X}_c$ implies $\xi(\tau_k) \in\mathcal{X}_{c,2}$.
\end{prop}
\begin{proof}
	Due to Assumption~\ref{asum_cont_clf}, \eqref{eq_non_mon_dec4} holds. Since $f_e(x) = 0$, it follows that $\xi_2(\tau_k+\tone) = \xi_1(\tau_k) = x(\tau_k)$ for $0 \leq \tone < \tau_{k+1}-\tau_k$ . \change{Furthermore, we}{We} have $\xi_1(t) = x(t)$ for all $t$ and  $\xi_2(\tau_{k+1}) = \xi_1(\tau_{k+1}) = x(\tau_{k+1})$ due to the structure of the DDS~\eqref{eq_cont_sys_comp}. The proposition follows then directly from \eqref{eq_v_desc} and \eqref{eq_v_desc_2}.
\end{proof}

\subsection{Alternative Characterization of the Convergence Criterion}
In this subsection, we present a sufficient condition for the convergence criterion \eqref{eq_cont_v_smaller} \change{}{for two arbitrary time points}. 
\change{The condition can be used between two arbitrary time points, which can e.g. be chosen as two sampling times of the PETC.}{} 
\change{It}{The condition} does not require explicit knowledge  of $S(t,x_0)$ and will turn out to be useful later in the PETC design.
\begin{prop}
	\label{prop_cont_conv}
	Consider two constants $C_1$, $C_2 ~\in \mathbb{R}_{\geq 0}$,  and $S(t,x_0)$ defined by \eqref{eq_cont_def_s}. If 
\change{	\begin{align}
	&C_1   	\leq C_2- \tone \sigma  \gamma(C_2) \label{eq_cont_conv}
	\end{align}}{$C_1\leq C_2- \tone \sigma  \gamma(C_2)$}
	and  $C_2 \leq S(s,x_0)$  	holds for $  \tone,s \in \mathbb{R}_{\geq 0}$ ,
	then 
	\change{\begin{equation}
	\label{eq_cont_conv_s}
	C_1 \leq S(s+\tone,x_0).
	\end{equation}}{$C_1 \leq S(s+\tone,x_0).$}
\end{prop}

\begin{proof}
	Denote by $s+t_1$ with $t_1 \geq 0$ the first time after $s$, for which $S(s+t_1,x_0) =C_2$. 
	Then, we notice that $C_1 \leq S(s+\tone,x_0)$ if $0 \leq \tone\leq t_1$ by assumption.
	If $\tone > t_1$, then 
	\begin{align*}
	S(s+\tone,x_0)\overset{\eqref{eq_cont_def_s}}{=}& 	S(s+t_1,x_0) - \int_{t_1}^{\tone} \sigma \gamma(S(s+\theta,x_0)) d\theta \\
	\geq  &C_2  - \int_{t_1}^{\tone} \sigma\gamma(C_2)  d\theta\\
	= &C_2  - (\tone-t_1) \sigma  \gamma(C_2)	\geq C_1. 
	\end{align*}
\end{proof}

\subsection{A Time Dependent Bound on the Lyapunov Function}
In this subsection, we recap from \cite{proskurnikov2018lyapunov} how a time dependent and state independent upper bound on the time derivative of $V_\CLF(x(t))$ can be computed.
\change{ (identical assumptions as in \cite{proskurnikov2018lyapunov} can also be found in the conference version \cite{pros2018conf}). }{}
This bound is used in \cite{proskurnikov2018lyapunov} to compute a lower bound on the $\sigma$-MASP for which the decrease of $V_\CLF(x(t))$ with a chosen convergence rate according to \eqref{eq_cont_cont_conv} can be guaranteed. 

We will show in Section~\ref{sec_main_res} how an improved lower bound on the $\sigma$-MASP \change{}{with convergence criterion \eqref{eq_cont_v_smaller}} for the PETC can be obtained using the same upper bound on the time derivative of $V_\CLF(x(t))$   \change{such that asymptotic stability and satisfaction of the convergence criterion \eqref{eq_cont_v_smaller} can be guaranteed based on non-monotonic Lyapunov functions}{and non-monotonic Lyapunov functions}. The\change{considered}{} scenario is described by the following Assumptions.   
\begin{asum}
\label{as_cont_pro1}

(cf.  Assumption~1 and 2 from \cite{proskurnikov2018lyapunov}
) For the chosen $c\in\mathbb{R}_{>0}$, there is a finite Lipschitz constant $L_{1,c}$ satisfying
\begin{equation*}
L_{1,c} \overset{\triangle}{=} \underset{x_1,x_2,x_3\in\mathcal{X}_c, x_1 \neq x_2}{\sup} \frac{\norm{f(x_1,~\kappa(x_3))-f(x_2,~\kappa(x_3))}}{\norm{x_1-x_2}}.
\end{equation*} 
\end{asum}
This assumption implies, that the difference of the system dynamics for two points $x_1,x_2 \in \mathcal{X}_c$ with $u$ chosen as the feedback  $\kappa(x_3)$ from \eqref{eq_cont_fb} for an arbitrary third point $x_3 \in \mathcal{X}_c$ can be bounded by a Lipschitz constant $L_{1,c}$ as $ L_{1,c} \norm{x_1-x_2}$. This assumption needs to hold only on the considered level set of $V_\CLF(x)$, that is defined by $\mathcal{X}_c$ and is thus not too restrictive. 

\begin{asum}
	\label{as_cont_pro2}

	(	cf.	Assumption~3 in \cite{proskurnikov2018lyapunov}
)   For the chosen  $c\in\mathbb{R}_{>0}$, there is a finite Lipschitz constant $L_{2,c} \in \mathbb{R}$ satisfying
	\begin{equation*}
	L_{2,c} \overset{\triangle}{=} \underset{x_1,x_2\in\mathcal{X}_c, x_1 \neq x_2}{\sup} \frac{\norm{V_\CLF^{'}(x_1) -V_\CLF^{'}(x_2) }}{\norm{x_1-x_2}}.
	\end{equation*}
\end{asum}

\change{This assumption}{Assumption~\ref{as_cont_pro2}} imposes smoothness requirements on $V_\CLF(x)$ and holds e.g. when $V_\CLF(x)$ is twice continuously differentiable. 
\begin{asum}
	\label{as_cont_pro3}

	 (cf. Assumption~4 resp. Lemma~1 from \cite{proskurnikov2018lyapunov}) For the chosen $c\in\mathbb{R}_{>0}$, there is a positive definite function $M_c:\mathbb{R}^n \rightarrow \mathbb{R},$ bounded on $\mathcal{X}_c$, satisfying for all $\in\mathcal{X}_c$
	\begin{align*}
	&\norm{V_\CLF^{'}(x) } \norm{f(x,~\kappa(x))} + \norm{f(x,~\kappa(x))}^2\\
	 \leq& M_c(x) \abs{\partder{x} f(x,~\kappa(x))}	.
	\end{align*}

\end{asum}

This assumption excludes systems with solutions $x(t)$, that are fast oscillating with fast changing continuous-time control $\kappa(x(t))$. For such systems, no finite sampling rate would be sufficient to maintain the descent of $V_\CLF(x)$ below a chosen bound. A detailed discussion is given in \cite{proskurnikov2018lyapunov}. 

\change{In order to obtain a time dependent but state independent bound on $V_\CLF(x(t))$ after a successful transmission of the controller, }{}
We consider now an additional Cauchy problem $\dot{\tilde{x}}(t) = f(\tilde{x},u_*), \tilde{x}(0) = \tilde{x}_0 \in \mathcal{X}_c$  for the chosen $c \in \mathbb{R}_{>0}$ and some $u_*$. We define $t^*$ as the first time after $t = 0$ for which $V_\CLF(\tilde{x}(t)) \geq c$  and $\triangle_*(\tilde{x}_0,u_*) = \left[0, t^*\right]$. We obtain the following upper bound on the time derivative of $V_\CLF(\tilde{x}(t))$ that was derived in \cite{proskurnikov2018lyapunov}.

\begin{coro}
	\label{coro_cont_lyap}
	(deviated from Corollary~4 from \cite{proskurnikov2018lyapunov})	Let Assumptions~\ref{as_cont_pro1} and \ref{as_cont_pro2} hold for the  chosen $c\in\mathbb{R}_{>0}$. Let $\tilde{x}_0,~\tilde{x}_1 \in \mathcal{X}_c$ , $~u_* = \kappa(\tilde{x}_1)$, $t\in \triangle_*(\tilde{x}_0,u_*) \cap \left[0, (1+2L_{1,c})^{-1}\right]$ and $\tilde{x}(t)$ be the solution of $\dot{\tilde{x}}(t) = f(\tilde{x}(t),u_*), \tilde{x}(0) = \tilde{x}_0$. Then, 
	\begin{align}
	&\abs{\timeder{x(t)}{u_*} - \timeder{\tilde{x}_0}{u_*}} \nonumber \\
	\leq& \sqrt{t} \mu_c \left(\norm{ V_\CLF^{'}(\tilde{x}_0)}  \norm{f(\tilde{x}_0,u_*)} + \norm{f(\tilde{x}_0,u_*)}^2 \right),
	\end{align}
	where
\change{	\begin{equation}
	\mu_c \triangleq \sqrt{e} \max \left\lbrace L_{1,c},L_{2,c}(1+L_{1,c} \sqrt{e})\right\rbrace.
	\end{equation}}{$\mu_c \triangleq \sqrt{e} \max \left\lbrace L_{1,c},L_{2,c}(1+L_{1,c} \sqrt{e})\right\rbrace.$}
\end{coro} 	
	\begin{proof}
		See
		\cite{proskurnikov2018lyapunov}.
	\end{proof}
The upper bound on the time derivative of $V_\CLF(\tilde{x}(t))$ that was computed in Corollary~\ref{coro_cont_lyap} can now be used in different ways in order to design a PETC mechanism.
In \cite{proskurnikov2018lyapunov}, it is shown how a lower bound on the $\sigma$-MASP  and a PETC method can be designed such that \eqref{eq_cont_cont_conv} is guaranteed for all times. In order to exploit non-monotonic Lyapunov functions, we present in the next section an alternative approach based on integrating the bound on the time derivative of $V_\CLF(\tilde{x}(t))$. We will show with this integrated bound, how an improved lower bound on the $\sigma$-MASP and a PETC mechanism can be designed\change{, such that the conditions from Proposition~\ref{prob_michel} hold and satisfaction of the convergence criterion  \eqref{eq_cont_v_smaller} is guaranteed
	}{}.

\section{Main Results}
\label{sec_main_res}
Now, we proceed to our main results that are the construction of a lower bound on the $\sigma$-MASP and the design of a PETC mechanism such that asymptotic stability of the origin of the DDS~\eqref{eq_cont_sys_comp} and satisfaction of the convergence criterion \eqref{eq_cont_v_smaller} are guaranteed. \change{if the  PETC mechanism is used periodically with sampling period $ h $ chosen according to the $\sigma$-MASP bound or smaller.}{} 
\subsection{A Lower Bound on the $\sigma$-MASP}
In this subsection, we tackle the problem of constructing a lower bound on the $\sigma$-MASP such that asymptotic stability of the origin of the DDS~\eqref{eq_cont_sys_comp} and satisfaction of the convergence criterion \eqref{eq_cont_v_smaller}  are guaranteed for periodic triggering with sampling period $h$ chosen according to the $\sigma$-MASP bound if at most $m$ successive packets are lost. The sampling period $h$ will also be used subsequently for the PETC mechanism. 
We assume that a Lyapunov function and a controller for the continuous-time system are designed to satisfy Assumption~\ref{asum_cont_clf}. Then,
 the following Lemma can be used to construct a lower bound on the $\sigma$-MASP, here denoted by $h_{\sigma\text{\normalfont -MASP}}$, and to choose the sampling period $ h $.

\begin{lemma}
	\label{lem_cont_stab2}
	Let Assumptions~\ref{asum_cont_clf}-\ref{as_cont_pro3} hold for the chosen $c\in\mathbb{R}_{>0}$. Assume  system \eqref{eq_sys_cont} is used with controller \eqref{eq_cont_fb}, prediction \eqref{eq_cont_est}, $f_e(x) = 0$, and with $x(\tau_k)\in\mathcal{X}_c$ for some $\tau_{k}\in\mathcal{T}$. 
	Let the next successful transmission take place at a time $\tau_{k+1} = \tau_k + j  h $ for some $j \in \left\lbrace 1,...,m+1\right\rbrace$ and $(m+1) h  \leq h_{\sigma\text{\normalfont -MASP}}$ with 
	\begin{equation}
	\label{eq_h_max}
		h_{\sigma\text{\normalfont -MASP}} = \min \left\lbrace \left(\frac{3 (1-\sigma)}{2 \mu_c M_{\max,c}}\right)^2, (1+2L_{1,c})^{-1}\right\rbrace,
	\end{equation}
	 where 
	 	$M_{\max,c} = \underset{x \in \mathcal{X}_c}{\sup} M_c(x)$
	 and $\sigma \in \left(0,1\right)$. Then, 
\eqref{eq_v_desc} and \eqref{eq_v_desc_2} hold for $V_\CLF$  on $\mathcal{X}_{c}$
with $\gamma_2 = \sigma  \gamma $. Furthermore, if $V_\CLF(x(\tau_k)) \leq S(\tau_k,x_0)$, then the convergence criterion \eqref{eq_cont_v_smaller} holds for $t \in \left[\tau_k,\tau_{k+1}\right]$ and $V_\CLF(x({\tau_{k+1}})) \leq S(\tau_{k+1},x_0)$.

\end{lemma} 

 \begin{proof}
 	\arxiv{The proof is omitted due to spacial limitations. It can be found in the accepted version \cite{hertneck2019archive}.}{The proof is given in Appendix~\ref{append_a}.}
 \end{proof}
 
\begin{rema}
	If a transmission is triggered periodically with a sampling period of $h$, then Proposition~\ref{prob_michel}, Lemma~\ref{lem_cont_stab2} and Proposition~\ref{prop_help_eq}  can be used to guarantee asymptotic stability of the origin of the DDS~\eqref{eq_cont_sys_comp} and satisfaction of the convergence criterion \eqref{eq_cont_v_smaller} even \change{though}{if} up to $m$ successive packets may be lost. For $m=0$, $h$ may be equal to the $\sigma$-MASP bound $h_{\sigma\text{\normalfont -MASP}}$.
\end{rema}

\begin{coro}
	If
\change{	\begin{equation*}
	\left(\frac{3 (1-\sigma)}{2 \mu_c M_{\max,c}}\right)^2 \leq (1+2L_{1,c})^{-1}
	\end{equation*}}{$\left(\frac{3 (1-\sigma)}{2 \mu_c M_{\max,c}}\right)^2 \leq (1+2L_{1,c})^{-1}$}
	holds,  we obtain a bound on the $\sigma$-MASP that is at least $\frac{9}{4}$ times larger than the bound from \cite{proskurnikov2018lyapunov} that is given in the least conservative case according to \cite[Lemma~3]{proskurnikov2018lyapunov} by 
\change{	\begin{equation*}
	h_{\sigma\text{\normalfont -MASP},\text{\normalfont \cite{proskurnikov2018lyapunov}}} =  \min \left\lbrace \left(\frac{ (1-\sigma)}{ \mu_c M_{\max,c}}\right)^2 ,(1+2L_{1,c})^{-1} \right\rbrace
	\end{equation*}}{$h_{\sigma\text{\normalfont -MASP},\text{\normalfont \cite{proskurnikov2018lyapunov}}} =  \min \left\lbrace \left(\frac{ (1-\sigma)}{ \mu_c M_{\max,c}}\right)^2 ,(1+2L_{1,c})^{-1} \right\rbrace$}
	while \change{still}{} guaranteeing the same worst case average convergence rate.
\end{coro}

\subsection{A Robust Dynamic PETC Mechanism}
Now, we are ready to present a dynamic PETC mechanism, that guarantees asymptotic stability and satisfaction of the convergence criterion \eqref{eq_cont_v_smaller} 
 despite packet loss. 
 
 The trigger mechanism is \change{summarized in}{given by} Algorithm~\refalgoo ~which can be viewed as a nonlinear counterpart to  Algorithm~1 from \cite{linsenmayer2018event2}.  
 The main idea is to decide whether a transmission is necessary based on an upper bound on the evolution of the systems Lyapunov function derived from Corollary~\ref{coro_cont_lyap} and on a time-varying trigger rule.
 The bound on the evolution of the Lyapunov function (denoted by $\sigma_z$ in Algorithm~\refalgoo) changes dynamically with the number of failed transmissions since the last successful transmission ($\bar{m}$ in Algorithm~\refalgoo), and  replaces the prediction based on an exact discretization, that was used in \cite{linsenmayer2018event2} to determine if triggering is necessary for linear systems. The guarantees for asymptotic stability and satisfaction of the convergence criterion~\eqref{eq_cont_v_smaller} rely on the choice of the sampling period according to Lemma~\ref{lem_cont_stab2}. Based on Algorithm~\refalgoo, we can state the following theorem.

\begin{table}[tb]
	\begin{algorithm}[H]
		\caption{Dynamic triggering mechanism at $z\in\mathbb{N}_0$.}
		\label{algo_cont_dyn_trig}
	\begin{algorithmic}[1]
		\If{$z = 0$}
		\State $i^\text{\normalfont ref} \leftarrow 0$,  $x^\text{\normalfont ref} \leftarrow x_0$, $\bar{m} \leftarrow 0$,  $u_* \leftarrow \kappa(x_0)$
		\State $V^\text{\normalfont ref} \leftarrow V_{\CLF}(x^\text{\normalfont ref})$ with $V_{\CLF}$ according to Assumption~\ref{asum_cont_clf}
		\State send $x_0$ over the network (successful by assumption)
		\Else
	
		\State $\sigma_z \leftarrow  h  (m-\bar{m}+1) \timeder{x(zh)}{u_*} +\frac{2}{3} ( h (m-\bar{m}+1))^{3/2} \mu_c \left( \norm{\partder{x(zh)}} \norm{f(x(zh),u_*)} + \norm{f(x(zh),u_*)}^2\right) $

		\If {$z-i^\text{\normalfont ref} > \nu$ \textbf{or} 			
			$V_\CLF(x(zh))+\sigma_z \geq V^\text{\normalfont ref} - (z-i^\text{\normalfont ref}+ m -\bar{m} +1)  h  \sigma\gamma \left( V^\text{\normalfont ref}\right)$} \label{line_trigger_cond}
		\State send $x(zh)$ over the network and wait for acknowledgment 
		\If {transmission is successful}
		\State $i^\text{\normalfont ref} \leftarrow z$, $V^\text{\normalfont ref} \leftarrow V_{\CLF}(x(zh))$, $\bar{m} \leftarrow 0$, $u_* \leftarrow \kappa(x(zh))$ 
		\Else
		\State $\bar{m} \leftarrow \bar{m}+1$
		\EndIf
		\Else
		\State no transmission of $x(zh)$ necessary
		\EndIf

		\EndIf
	\end{algorithmic}
	\end{algorithm}
\end{table}
 \begin{theo}
 	\label{theo_cont_stab}
 	Let Assumptions~\ref{asum_cont_clf}-\ref{as_cont_pro3} hold on $\mathcal{X}_c$ for the chosen $c \in \mathbb{R}_{>0}$. Assume  system \eqref{eq_sys_cont} is used with controller \eqref{eq_cont_fb}, prediction \eqref{eq_cont_est}, with $f_e(x) = 0$ and with $x(0) \in \mathcal{X}_c$. Assume furthermore that necessary transmissions are detected with the trigger mechanism specified by Algorithm~\refalgoo ~that is evaluated periodically with a sampling period $ h $ chosen as in Lemma~\ref{lem_cont_stab2}, $\sigma \in  \left(0, 1\right)$, $\nu \in \mathbb{N}$ \change{}{(arbitrary large)} and that current state information is received successfully at $\tau_0 = 0$. Then, the origin of the DDS~\eqref{eq_cont_sys_comp} is locally asymptotically stable with region of attraction $x_0 \in \mathcal{X}_{c}$ and the convergence criterion \eqref{eq_cont_v_smaller} is satisfied. 
 \end{theo}
\begin{proof}
		\arxiv{The proof is omitted due to spacial limitations. It can be found in the accepted version \cite{hertneck2019archive}.}{The proof is given in Appendix~\ref{append_b}.}
\end{proof}
 \begin{rema}
 	Algorithm~\refalgoo~and the lower bound on the $\sigma$-MASP can in principle still be used to guarantee stability if $\gamma$ in Assumption~\ref{asum_cont_clf} is only positive definite instead of being of class $\mathcal{K}$. However, Proposition~\ref{prop_cont_conv} does not hold \change{for this case}{then}.
 \end{rema}
 
 \begin{rema}
 	Instead of the dynamic trigger rule in line~\ref{line_trigger_cond} of Algorithm~\refalgoo, 	different trigger rules can be easily incorporated in the Algorithm. Whilst the trigger rule in Algorithm~\refalgoo ~ leads to a bound on $V_\CLF(x(\tau_k))$, that decreases piecewise linearly (between elements  of $\mathcal{T}$), different rules can lead to bounds with different behavior. This can be demonstrated easily for exponentially stabilizable systems, i.e. with $\gamma(V) = KV$ for some $K \in \mathbb{R}_{>0}$. For such systems, we obtain
$	S(t,x_0) = e^{-K\sigma t} V_\gamma(x(0)).$ 
 	Then, we can use   		
 	\vspace{-2.5mm}
 	\begin{table}[H]

 		\footnotesize
 		\begin{algorithmic}[1]	 		
 			\setcounter{ALG@line}{\getrefnumber{line_trigger_cond}-1}
 			\State \textbf{if} {$z-i^\text{\normalfont ref} > \nu$ \textbf{or} 
 					
 				$V_\CLF(x(zh))+\sigma_z \geq  e^{-K\sigma(z-i^\text{\normalfont ref}+ m -\bar{m} +1)  h  } V^\text{\normalfont ref}$}	\textbf{then}
 		\end{algorithmic}
 		\normalsize
 			\vspace{-5.5mm}
 	\end{table} 
 	 \noindent as trigger rule for Algorithm~\refalgoo. If the DDS~\eqref{eq_cont_sys_comp} is controlled using Algorithm~\refalgoo ~with this trigger rule, one can show that \eqref{eq_cont_v_smaller} holds and guarantee  stability of the DDS~\eqref{eq_cont_sys_comp} using a similar argumentation as in the proof of Theorem~\ref{theo_cont_stab}. Differently shaped bounds are possible as well.
 \end{rema}

\begin{rema}
	
	In Algorithm~\refalgoo, an adaptive trigger rule that raises the number of triggered transmissions if the network load is low, similar to the mechanism from \cite{linsenmayer2018event2}, can be easily included \change{To do this, an additional term can be incorporated in the trigger rule, that can be e.g. chosen as }{by modifiying the trigger rule as} 	\vspace{-2.5mm}
	\begin{table}[H]

	\footnotesize
	\begin{algorithmic}[1]
	\setcounter{ALG@line}{\getrefnumber{line_trigger_cond}-1}
		\State \textbf{if} {$z-i^\text{\normalfont ref} > \nu$ \textbf{or} 	
			$V_\CLF(x(zh))+\sigma_z $
			
		$	\geq V^\text{\normalfont ref} - (z-i^\text{\normalfont ref}+ m -\bar{m} +1)  h  (\sigma +c_n(z)) \gamma \left( V^\text{\normalfont ref}\right)$} \textbf{then}
	\end{algorithmic}
	\normalsize
	\vspace{-5.5mm}
	\end{table} 
\noindent	for an adaptive $c_n(z) \geq 0$ that can depend on the state of the communication network. Obviously, this rule leads always to a trigged transmission if the original rule from Algorithm~\refalgoo ~would lead to a triggered transmission.
\end{rema}


\section{Numerical Example}
\label{sec_example}
In this section, an academic pendulum example from \cite{boubaker2013inverted} is employed to demonstrate the proposed PETC method. The system dynamics are given by 
\begin{equation}
\begin{pmatrix}
\dot{x}_1(t) \\
\dot{x}_2(t) 
\end{pmatrix} = \begin{pmatrix}
x_2(t)\\
(sin(x_1(t)) - u(t) cos(x_1(t))) \omega_0
\end{pmatrix}
\end{equation}
with pendulum angle $x_1$, angular velocity $x_2$, input $u$, that is a force that acts on the mass center of the pendulum, and a constant $\omega_0$. For $\omega_0 = 0.1$, $\kappa(x) = \frac{31.6x_1+40.4x_2+sin(x_1)}{cos(x_1)}$ and $V_\CLF(x) = 1.278 x_1^2 + 0.632 x_1 x_2 + 0.404 x_2^2 $, $V_\CLF$ and the resulting DDS~\eqref{eq_cont_sys_comp} satisfy\change{\footnote{The singularity for $cos(x_1) = 0$ is ruled out by the choice for $c$.}}{} for $c = 0.258$ and $\sigma = 0.35$ the assumptions of Theorem~\ref{theo_cont_stab} with 
$\left(\frac{3 (1-\sigma)}{2 \mu_c M_{\max,c}}\right)^2 = 2.77\cdot 10^{-5}$ and $(1+2L_1)^{-1} \approx \frac{1}{4.3}$. Thus, Algorithm~\refalgoo ~can be used to stabilize the pendulum, and we obtain a $\sigma$-MASP bound that is at least  $\frac{9}{4}$ times higher than the $\sigma$-MASP bound from \cite{proskurnikov2018lyapunov} whilst guaranteeing the same average convergence rate.

We consider a network with  uniformly distributed packet dropouts with $m=1$ and obtain $h = 1.38\cdot 10^{-5}$.  Figure~\ref{fig_example} shows state trajectories, input trajectories and $V_\CLF(\xi(t))$ for $ \xi = \left[x_1, x_2,\hat{x}_1,\hat{x}_2\right]^\top$ and $x_0 = [0.43,0]^\top$. 
The average time between two successful transmissions of the  controller is $0.48 s$. Thus, the number of triggered transmission is reduced significantly if the proposed PETC mechanism is used in comparison to periodic time-triggered sampling.

\begin{figure}[tb]
	 \centering
	\includegraphics[width = \linewidth]{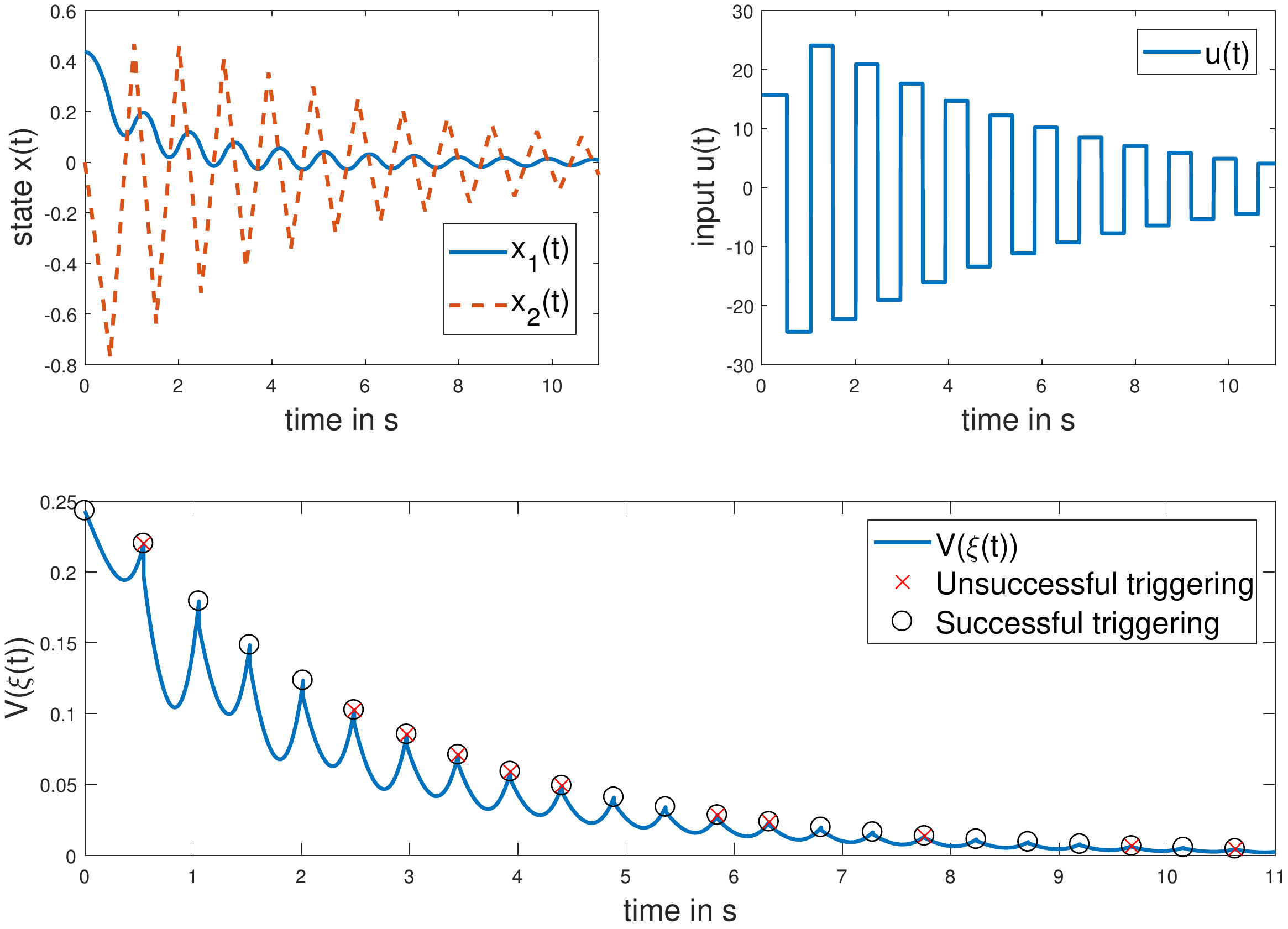}
	\caption{Evolution of system state $x_1(t)$ and $x_2(t)$ (above left), input $u(t)$ (above right) and $V(\xi(t))$ (below) for the pendulum controlled by Algorithm \refalgoo.}
	\label{fig_example}
\end{figure}

\section{Conclusion}
\label{sec_conclusion}
In this paper, we presented a PETC mechanism that can guarantee asymptotic stability and a chosen average convergence rate while reducing the amount of communication for a wide class of nonlinear systems. \change{,assuming existence of a stabilizing controller and a Lyapunov function for the continuous-time system.}{}The theoretical derivations rely on non-monotonic Lyapunov functions that have shown to be a powerful tool for the PETC design. The proposed PETC mechanism is robust to packet loss if a bound on the number of successive lost packets is known. In addition, a method to construct an improved lower bound on the $\sigma$-MASP based on a result from \cite{proskurnikov2018lyapunov} was presented.

Whilst the proposed  PETC mechanism shows a nice behavior in simulations, there are still open points for future research. For example, a modification  in order to guarantee robustness to disturbances and model uncertainties seems to be a natural extension of the proposed PETC mechanism. 
Moreover, using a computable model-based prediction for $\hat{x}$ at the actuator, comparable to the result for linear systems from \cite{Heemels2013}, may be an interesting approach to obtain a further reduction of the number of transmissions that are triggered by the controller.
\bibliography{sources}

\arxiv{}{
  \appendix
  \subsection{Proof for Lemma~\ref{lem_cont_stab2}}
  \label{append_a}
\begin{proof}

	Since $f_e(x) = 0$, we can consider the system $\dot{\tilde{x}} = f(\tilde{x},u_*),~\tilde{x}(0)  = x(\tau_k), u_*  = \kappa(x(\tau_k))$ and have $V_\CLF(\tilde{x}( \tone )) = V_\CLF(x(\tau_k+\tone)) $ for $0 \leq \tone \leq \tau_{k+1} - \tau_k $. Due to Corollary~\ref{coro_cont_lyap} and Assumption~\ref{as_cont_pro3}, we know for $0 < \tone \leq \min \left\lbrace \tau_{k+1} - \tau_k, \left(1+2L_{1,c} \right)^{-1} , t^*\right\rbrace $, that
\change{	\begin{align*}
	&\timeder{\tilde{x}( \tone )}{u_*} \\
	\leq& \timeder{\tilde{x}(0)}{u_*} + \sqrt{\tone} \mu_c M_c(\tilde{x}(0)) \abs{ \timeder{\tilde{x}(0)}{u_*}}
	\end{align*}}{$\timeder{\tilde{x}( \tone )}{u_*} \\
	\leq \timeder{\tilde{x}(0)}{u_*} + \sqrt{\tone} \mu_c M_c(\tilde{x}(0)) \abs{ \timeder{\tilde{x}(0)}{u_*}}$}
	holds, and thus since $\timeder{\tilde{x}(0)}{u_*}\leq -\gamma(V_\gamma(\tilde{x}(0))) \leq 0$ due to Assumption~\ref{asum_cont_clf} and with $M_c(\tilde{x}(0)) \leq M_{\max,c}$,
	\begin{align}
	\nonumber
	\timeder{\tilde{x}( \tone )}{u_*} \leq& \timeder{\tilde{x}(0)}{u_*} (1-\sqrt{ \tone } \mu_c M_c(\tilde{x}(0)))\\
	\nonumber
	\leq&  - \gamma(V_\CLF(\tilde{x}(0))) (1-\sqrt{ \tone } \mu_c M_{\max,c})
	\end{align}
	holds. By a time integration, we obtain,  
 that
	\begin{align*}
	&V_\CLF(\tilde{x}( \tone ))   \\
	\leq& V_\CLF(\tilde{x} (0)) -\gamma(V_\CLF(\tilde{x}(0)))\int_{0}^{ \tone } 1-\sqrt{ \theta } \mu_c M_{\max,c} d\theta \\
	\leq& V_\CLF(\tilde{x} (0))- \gamma(V_\CLF(\tilde{x}(0))) ( \tone -\frac{2}{3}  \tone ^{3/2} \mu_c M_{\max,c}),
	\end{align*}
	and hence 
	\begin{align}	
	\nonumber &V_\CLF({x}(\tau_k+\tone)) \\
	\leq& V_\CLF({x} (\tau_k))  -\gamma(V_\CLF({x}(\tau_k))) (\tone-\frac{2}{3} \tone^{3/2} \mu_c M_{\max,c})  \label{eq_lyap_bound}
	\end{align}
	 holds. If \eqref{eq_lyap_bound} holds, then we observe that
	\begin{equation}
	\label{eq_cont_lem_bound}
	V_\CLF({x}(\tau_k+\tone)) \leq V_\CLF(x(\tau_k)) - \tone \sigma \gamma(V_\CLF(x(\tau_k)))< c
	\end{equation}
	 holds if
	\begin{align}
	\tone(1-\sigma-\frac{2}{3}\sqrt{\tone} \mu_c M_{\max,c}) \geq 0. \label{eq_cont_inc}
	\end{align}
	The left hand side of \eqref{eq_cont_inc} is continuous and has zeros only at $\tone=0$ and at $\tone = \left(\frac{3(1-\sigma)}{2 \mu_c M_{\max,c}} \right)^2 \overset{\eqref{eq_h_max}}{\geq} 	h_{\sigma\text{\normalfont -MASP}}$. The second derivative of the left-hand side of \eqref{eq_cont_inc} w.r.t. $r$ is strictly negative for $\tone>0$ and  thus \eqref{eq_cont_lem_bound} holds for $\tone \in \left(0,	h_{\sigma\text{\normalfont -MASP}}\right]$ if \eqref{eq_lyap_bound} holds. This implies directly, that $r \leq h_{\sigma\text{\normalfont -MASP}} < t^*$ and hence that \eqref{eq_v_desc} and \eqref{eq_v_desc_2} with $\gamma_2 = \sigma \gamma$ hold for $\tau_k$, if $\tau_{k+1} \leq \tau_k + 	h_{\sigma\text{\normalfont -MASP}} \leq \tau_k+(1+2L_{1,c})^{-1}$ what is ensured since $\tau_{k+1} = \tau_k+j h $ for some  $j \in \left\lbrace 1,...,m+1 \right\rbrace$ and $(m+1) h  \leq	h_{\sigma\text{\normalfont -MASP}}$. 	
	To show that the convergence criterion \eqref{eq_cont_v_smaller} holds, we can use \eqref{eq_cont_lem_bound} and Proposition~\ref{prop_cont_conv} with $C_1 =  V_\CLF({x}(\tau_k+\tone)) $, $C_2 = V_\CLF(x(\tau_k)) \leq S(\tau_k,x_0)$, $s = \tau_k$, and every fixed $r\in \left(0,\tau_{k+1}-\tau_k\right]$   to show that
	\begin{align}
	V_\CLF(x(\tau_k + \tone )) \leq S(\tau_k+\tone,x_0) \label{eq_cont_s_better}
	\end{align}
	holds for $\tone \in \left(0, \tau_{k+1}-\tau_k \right]$. This is a stronger result than \eqref{eq_cont_v_smaller} since $ S(t,x_0)$ is monotonically decreasing in $t$ and implies thus also that  \eqref{eq_cont_v_smaller} holds for $t\in\left[\tau_k, \tau_{k+1}\right]$.

\end{proof}
\subsection{Proof for Theorem~\ref{theo_cont_stab}}
\label{append_b}
 \begin{proof}
  We show that the conditions of Proposition~\ref{prob_michel} hold for the DDS~\eqref{eq_cont_sys_comp} and $V(\xi) = V_\CLF(\xi_1) + V_\CLF(\xi_2)$, such that we can conclude asymptotic stability. Furthermore, we prove simultaneously that the convergence criterion \eqref{eq_cont_v_smaller} holds. 
%
 	
 	The first transmission is always successful, i.e. $\tau_0 = 0$. Moreover, we note that a transmission is triggered if the number of periods since the last successful transmission ($z-i^\text{\normalfont ref}$ in Algorithm~\refalgoo) exceeds a bound $\nu$. Additionally, the time between two successful transmissions is lower bounded by $ h $ and thus,  \eqref{eq_non_mon_dec3} holds. 	
 	To show that \eqref{eq_non_mon_dec4}, \eqref{eq_non_mon_desc1} and \eqref{eq_non_mon_desc2} hold, we show that \eqref{eq_v_desc} and \eqref{eq_v_desc_2}  hold and use Proposition~\ref{prop_help_eq}.
 	Due to Proposition~\ref{prop_help_eq}, we know also that $x(\tau_k) \in\mathcal{X}_c$ implies $\xi(\tau_k) \in\mathcal{X}_{c,2}$.
 	
 	 We distinguish between two cases. For each sampling time with successful transmission  $\tau_k$, there is either at least one sampling time where no transmission is necessary according to Algorithm~\refalgoo ~until the next sampling 
 	 time with successful transmission  $\tau_{k+1}$, or there is none. 
 	
 	If there is none and $V_\CLF(x(\tau_k)) \leq S(\tau_k,x_0)$, then we know due to Lemma~\ref{lem_cont_stab2} that \eqref{eq_cont_v_smaller}, \eqref{eq_v_desc} and \eqref{eq_v_desc_2} hold between $\tau_k$ and $\tau_{k+1}$ and $V_\CLF(x({\tau_{k+1}})) \leq S(\tau_{k+1},x_0)$, because the next successful transmission takes place within the next $m+1$ periods in this case by assumption.  Thus, it remains to show, that \eqref{eq_cont_v_smaller}, \eqref{eq_v_desc} and \eqref{eq_v_desc_2}  hold  between $\tau_k$ and $\tau_{k+1}$  and $V_\CLF(x({\tau_{k+1}})) \leq S(\tau_{k+1},x_0)$ if there are sampling instants where no transmission is necessary according to Algorithm~\refalgoo.

 	Therefore, we introduce the sequence of sampling instants between $\tau_k$ and $\tau_{k+1}$ where no transmission is necessary according to Algorithm~\refalgoo ~as $\left(l_p^{\tau_k}\right)_{p \in \left\lbrace 1,\dots,p_{\max}^{\tau_k} \right\rbrace}$ with some $p_{\max}^{\tau_k} \in \mathbb{N}$ . Thus, it holds that  $\tau_k < hl_p^{\tau_k} < \tau_{k+1}$ for all $ p \in \left\lbrace 1,\dots,p_{\max}^{\tau_k} \right\rbrace$.  
 	We denote the number of failed transmissions since $\tau_k$ at the $l_p^{\tau_k}$-th sampling time by $\bar{m}_{p}^{\tau_k} \leq m$ (this equals the value of  $\bar{m}$  in Algorithm~\refalgoo ~at $k = l_p^{\tau_k}$). In the sequel, we will omit the superscript $\tau_k$ for ease of notation.

 	Now, we show that for $l_{p_{\max}}$ there are  guarantees for \eqref{eq_cont_v_smaller}, \eqref{eq_v_desc} and \eqref{eq_v_desc_2} to hold between $\tau_k$ and $\tau_{k+1}$ if the next successful transmission takes place at one of the next $m -\bar{m}_{p_{\max}}+1$ sampling times after $hl_{p_{\max}}$ and that $V_\CLF(x(hl_{p_{\max}}+jh)) \leq S(hl_{p_{\max}}+jh,x_0)$ holds for $j\in\left\lbrace0,\dots,\min \left\lbrace m-\bar{m}_{p_{\max}}+1, \frac{\tau_{k+1}}{h} - l_{p_{\max}}\right\rbrace \right\rbrace$ in this case. 

 	To do this, we consider  first an arbitrary	$l_p$, for which we assume to know that \eqref{eq_cont_v_smaller}, \eqref{eq_v_desc} and \eqref{eq_v_desc_2} hold  between $\tau_k$ and $\tau_{k+1}$ if the next successful transmission takes place at one of the next $m -\bar{m}_p$ sampling times after $hl_p$ and $V_\CLF(x(hl_p+jh)) \leq S(hl_p+jh,x_0)$ holds for $j\in\left\lbrace0,\dots ,  \min\left\lbrace m-\bar{m}_p,\frac{\tau_{k+1}}{h} - l_{p} \right\rbrace \right\rbrace$. We show now using the trigger rule that then \eqref{eq_cont_v_smaller}, \eqref{eq_v_desc} and \eqref{eq_v_desc_2} even hold  between $\tau_k$ and $\tau_{k+1}$ if the next successful transmission takes place at one of the next $m-\bar{m}_{p}+1$ sampling times and $V_\CLF(x(hl_p+jh)) \leq S(hl_p+jh,x_0)$ holds even for $j\in\left\lbrace0,\dots, \min \left\lbrace m-\bar{m}_p+1,\frac{\tau_{k+1}}{h} - l_{p}  \right\rbrace \right\rbrace$.  	
	We will later use this result iteratively from $l_1$ to $l_{p_{\max}}$ together with the fact that $hl_{p+1} = hl_p + (\bar{m}_{p+1}-\bar{m}_p+1)h$ to obtain the desired guarantees for $l_{p_{\max}}$.
 	We define the auxiliary function 
	\begingroup
	\begin{align}
		& V_\boundd(x(hl_p),\tone)  \nonumber\\
		:=&  V_\CLF(x(hl_p)) + \tone \timeder{x(hl_p)}{u_*} \nonumber \\
		&+ \frac{2}{3} \tone^{3/2} \mu_c \left( \norm{\partder{x(hl_p)}} \norm{f(x(hl_p),u_*)} \right. \nonumber \\
		&+ \left. \norm{f(x(hl_p),u_*)}^2\right). \label{eq_cont_bound}
	\end{align}
	\endgroup
 	Since no transmission was triggered at time $hl_p$, we know due to the trigger rule in Algorithm~\refalgoo ~and with $i^\text{\normalfont ref}h = \tau_k$ and $V^\text{\normalfont ref} = V_\CLF(x(\tau_k))$ that 
 	\begin{align}
 	&V_\CLF(x(hl_p ))+\sigma_z < V_\CLF(x(\tau_k))\nonumber \\ -& (hl_p  - \tau_k+(m-\bar{m}_p+1 )  h )  \sigma\gamma ( V_\CLF(x(\tau_k))).\label{eq_lyap_bound2}
 	\end{align}
 	Inserting 
 	 $\tone = h(m-\bar{m}_p+1)  $ in \eqref{eq_cont_bound}, we obtain with \eqref{eq_lyap_bound2}, $z=l_p$ and $\sigma_z$ from Algorithm~\refalgoo
 	 \begin{align}
	   &V_\boundd(x(hl_p),(m-\bar{m}_p+1)h)  <  V_\CLF(x(\tau_k))\nonumber \\
 	 -& (hl_p  - \tau_k+(m-\bar{m}_p+1 )  h ) \sigma \gamma(V_\CLF(x(\tau_k))). \label{eq_cont_v_bound_3}
 	 \end{align}
 	 The second derivative of $V_\boundd(x(hl_p),\tone)$ w.r.t. $r$ is positive for  $0 < \tone \leq (m-\bar{m}_p+1)  h $, and thus, $V_\boundd(x(hl_p),\tone)$ must have its maximum on the interval $0 \leq \tone \leq (m-\bar{m}_p+1)  h $ either at $\tone = 0$ or at $\tone = (m-\bar{m}_p+1)  h $, i.e. 
 	  	 \begin{align}
 	  	 & V_\boundd(x(hl_p),\tone ) \nonumber \\
 	  	 \leq&\max \left\lbrace  V_\CLF(x(hl_p)), V_\boundd(x(hl_p),(m-\bar{m}_p+1 )h  )  \right\rbrace < c \label{eq_v_bound_r}
 	  	 \end{align} 
 	  	 
 	  	 Now, we consider an auxiliary system starting at time $hl_p $ defined by $\dot{\tilde{x}} = f(\tilde{x},u_*),~\tilde{x}(0)  = x(hl_p), u_*  = \kappa(x(\tau_k))$ and have $V_\CLF(\tilde{x}( \tone )) = V_\CLF(x(hl_p+\tone)) $ for $0 \leq \tone \leq \tau_{k+1}-hl_p$. Then by Corollary~\ref{coro_cont_lyap}, we obtain using the same argumentation as in the proof of Lemma~\ref{lem_cont_stab2}, but without using Assumption~\ref{as_cont_pro3}  for $0\leq \tone\leq \min \left\lbrace \tau_{k+1}-hl_p, (1+2L_{1,c})^{-1},t^* \right\rbrace$ 
 	  	 \begin{equation}
 	  	 V_\CLF(x(hl_p+\tone))  =  V_\CLF(\tilde{x}( \tone )) \leq V_\boundd(x(hl_p),\tone). \label{eq_cont_bound2}
 	  	 \end{equation}
 	  	 Due to \eqref{eq_v_bound_r} and the choice of $h$ according to Lemma~\ref{lem_cont_stab2}, we know that ($m-\bar{m}_p+1)h \leq \min \left\lbrace t^*, (1+2L_{1,c})^{-1} \right\rbrace$.
	 	 Thus, \eqref{eq_cont_v_bound_3}-\eqref{eq_cont_bound2} ensure that \eqref{eq_v_desc} holds and \eqref{eq_v_desc_2} holds with $\gamma_2 =  \sigma \gamma$  if the next successful transmission takes place at  one of the next $m-\bar{m}_p+1$ sampling times after $hl_p$. To show that  \eqref{eq_cont_v_smaller} holds, we use \eqref{eq_cont_v_bound_3} and Proposition~\ref{prop_cont_conv} with $s = \tau_k$, $r = hl_p-\tau_k+(m-\bar{m}_p+1)h$, $C_1 = V_\boundd(x(hl_p),(m-\bar{m}_p+1 )  h)$ and $C_2  =V_\CLF(x(\tau_k)) \leq S(\tau_k,x_0) $ and obtain
	 	 $V_\boundd(x(hl_p),(m-\bar{m}_p+1 )  h) \leq S(hl_p+(m-\bar{m}_p+1 )  h,x_0).$
	 	 Hence, $V_\CLF(x(hl_p+jh)) \leq S(hl_p+jh,x_0)$ holds due to  \eqref{eq_cont_bound2} even for  $j\in\left\lbrace0,\dots,\min \left\lbrace m-\bar{m}_p+1, \frac{\tau_{k+1}}{h} - l_p \right\rbrace \right\rbrace$. Moreover, it holds due to the monotonicity of $S$, that 
	 	 $\max \left\lbrace  V_\CLF(x(hl_p)), V_\boundd(x(hl_p),(m-\bar{m}_p+1 )h  )  \right\rbrace \leq S(hl_p,x_0)$. Using additionally \eqref{eq_v_bound_r} and \eqref{eq_cont_bound2}, we observe that $V_\CLF(x(hl_p+r)) \leq S(hl_p,x_0)$ holds for $0\leq r \leq (m-\bar{m}_p+1)h$ and hence \eqref{eq_cont_v_smaller} holds between $\tau_k$ and $\tau_{k+1}$ if the next successful transmission takes place at one of the next $m-\bar{m}_p+1$ sampling times after $hl_p$.
	 	 
	 	 Since $hl_1 = \tau_k+(\bar{m}_1+1)h$,
	 	 we know from Lemma~\ref{lem_cont_stab2} that \eqref{eq_cont_v_smaller}, \eqref{eq_v_desc} and \eqref{eq_v_desc_2} hold between $\tau_k$ and $\tau_{k+1}$ if the next successful transmission takes place at one of the next  $m -\bar{m}_1$ sampling times and according  to \eqref{eq_cont_s_better} that $V_\CLF(x(\tau_k+jh)) \leq S(\tau_k+jh,x_0)$ holds for $j\in\left\lbrace0,\dots,\min \left\lbrace m+1, \frac{\tau_{k+1}-\tau_k}{h} \right\rbrace \right\rbrace$ and thus $V_\CLF(x(hl_1+jh)) \leq S(hl_1+jh,x_0)$ holds for $j\in\left\lbrace0,\dots,\min \left\lbrace  m-\bar{m}_1, \frac{\tau_{k+1}}{h} - l_1\right\rbrace \right\rbrace$. We can proceed now iteratively from $l_1$ to $l_{p_{\max}}$ using the above deviations for arbitrary $l_p$ and $hl_{p+1} = hl_p + (\bar{m}_{p+1}-\bar{m}_p +1)h$, to show that for $l_{p_{\max}}$, there are guarantees for \eqref{eq_cont_v_smaller}, \eqref{eq_v_desc} and \eqref{eq_v_desc_2} to hold between $\tau_k$ and $\tau_{k+1}$ if the next successful transmission takes place at one of the next $m -\bar{m}_{p_{\max}}+1$ sampling times after $hl_{p_{\max}}$ and that $V_\CLF(x(hl_p+jh)) \leq S(hl_p+jh,x_0)$ holds for $j\in\left\lbrace0,\dots,\min \left\lbrace m-\bar{m}_{p_{\max}}+1, \frac{\tau_{k+1}}{h} - l_{p_{\max}}\right\rbrace \right\rbrace$.

 	Finally, $\tau_{k+1}$ must be due to Assumption~\ref{asum_loss_bound} one of the next  $m-\bar{m}_{p_{\max}}+1$ sampling times after $hl_{p_{\max}}$. As a result, \eqref{eq_v_desc} and \eqref{eq_v_desc_2} and the convergence criterion \eqref{eq_cont_v_smaller} hold always between two sampling times with successful transmission
	 and we can use Propositions~\ref{prob_michel} and \ref{prop_help_eq}  to show asymptotic stability of the origin of the DDS~\eqref{eq_cont_sys_comp}.
 \end{proof} }


\end{document}